\documentclass[journal]{article}

\usepackage{graphicx,url}
\usepackage{dcolumn}
\usepackage{bm}
\usepackage{amsmath,amsfonts,amssymb}

\newtheorem{theorem}{\noindent{\it Theorem}}[section]

\newtheorem{lemma}[theorem]{\noindent{\it Lemma}}

\newtheorem{corollary}[theorem]{\noindent{\it Corollary}}
\newenvironment{proof}{\noindent{\it Proof:}}{$\hfill$ $\Box$\\ }
\newtheorem{example}{\noindent{\it Example}}[section]

\begin{document}

\title{Quantum codes derived from cyclic codes}

\author{Giuliano G. La Guardia
\thanks{Giuliano Gadioli La Guardia is with Department of Mathematics and Statistics,
State University of Ponta Grossa (UEPG), 84030-900, Ponta Grossa,
PR, Brazil. }}

\maketitle

\begin{abstract}
In this note, we present a construction of new nonbinary quantum
codes with good parameters. These codes are obtained by applying the
Calderbank-Shor-Steane (CSS) construction. In order to do this, we
show the existence of (classical) cyclic codes whose defining set
consists of only one cyclotomic coset containing at least two
consecutive integers.
\end{abstract}

\section{Introduction}

The class of cyclic codes is well-known in the literature
\cite{Macwilliams:1977,huffman03}. Recently, it has been extensively
employed in the construction of quantum codes
\cite{Calderbank:1998,Nielsen:2000,Ketkar:2006,laguardia:2009,laguardia:2011,laguardia:2014,Qian:2017}.
Let $q$ be a prime power. Recall that a $q$-ary quantum code
${\mathbb Q}$ of length $n$ is a $K$-dimensional subspace of the
$q^{n}$-dimensional Hilbert space ${({\mathbb C}^{q})}^{\otimes n}$,
where $\otimes n$ denotes the tensor product of vector spaces. If
$K=q^{k}$ we write $[[n, k, d]]_{q}$ to denote a $q$-ary quantum
code of length $n$ and minimum distance $d$. Let $[[n, k, d]]_{q}$
be a quantum code. The Quantum Singleton Bound (QSB) asserts that $k
+ 2d \leq n + 2$. If the equality holds then the code is called a
maximum distance separable (MDS) code. For more details on quantum
codes, the reader can consult \cite{Nielsen:2000,Ketkar:2006}.

In this note, we present constructions of new quantum codes by
applying the well-known CSS construction. In order to do this, we
show the existence of (classical) cyclic codes whose defining set
consists of only one cyclotomic coset containing at least two
consecutive integers. This fact induces the construction of quantum
codes with good parameters. More precisely, the code parameters are
good in the sense of the Singleton bound.

The paper is arranged as follows. In Section~\ref{sec2}, some
preliminaries results are provided. In Section~\ref{sec3}, we
present the contributions of the paper, i.e., constructions of new
quantum codes derived from (classical) cyclic codes. In
Section~\ref{sec4}, we give some examples of the new codes and we
compare the new code parameters with the ones available in the
literature. Finally, in Section~\ref{sec5}, the final remarks are
drawn.

\section{Basic Concepts}\label{sec2}

As usual, ${\mathbb F}_{q}$ represents a finite field with $q$
elements, where $q$ is a prime power. The parameters of a linear
code over ${\mathbb F}_{q}$ are denoted by $[n, k, d]_{q}$, where
$n$ is the code length, $k$ is the dimension and $d$ is the minimum
distance of the code. If $C$ is a linear code then $C^{\perp}$
denotes its (Euclidean) dual code. In this paper, we assume that
$\gcd(n, q) =1$ (simple root cyclic codes). The multiplicative order
of $q$ modulo $n$ is denoted by $m={\operatorname{ord}}_n(q)$. The
$q$-cyclotomic coset ($q$-coset for short) of $s$, modulo $n$, is
defined as $C_{s} =\{s, sq, \ldots , sq^{m_{s}-1} \}$, where $m_{s}$
is the smallest positive integer such that $sq^{m_{s}}\equiv s \bmod
n$. A primitive $n$th root of unity is denoted by $\alpha$.

Let $R_{n}={\mathbb F}_q [x]/(x^{n}-1)$ be the quotient ring of
polynomials modulo $(x^{n}-1)$. A cyclic code $C$ of length $n$ is a
non zero ideal in $R_{n}$. There exists only one polynomial $g(x)$
with minimal degree in $C$ such that $g(x)$ is a generator
polynomial of $C$, where $g(x)$ is a factor of $x^{n}-1$. The
dimension of $C$ equals $n - \deg(g(x))$. The dual of a cyclic code
is also cyclic. Recall the well-known BCH bound:

\begin{theorem}\label{CC}\cite[Pg. 201]{Macwilliams:1977}
(The BCH bound) Let $C$ be a cyclic code with generator polynomial
$g(x)$ such that, for some integers $b\geq 0$, $\delta\geq 1$, and
for $\alpha$ belongs to some extension field of ${\mathbb F}_{q}$,
we have $g({\alpha}^{b}) = g({\alpha}^{b+1})= \ldots =
g({\alpha}^{b+\delta-2})=0$, i.e., the code has a sequence of
$\delta-1$ consecutive powers of $\alpha$ as zeros. Then the minimum
distance of $C$ is, at least, $\delta$.
\end{theorem}

The Calderbank-Shor-Steane (CSS) quantum code construction is
well-known in the literature (see
\cite{Nielsen:2000,Calderbank:1998}). In the case of the classical
code to be dual-containing, the CSS construction reads as follows:

\begin{lemma}\cite[Lemma 17]{aly}\label{css}
If there exists a classical linear $[n, k, d]_{q}$ code $C$ such
that $C^{\perp} \subset C$, then there exists an $[[n, 2k - n, \geq
d]]_{q}$ stabilizer code that is pure to d.
\end{lemma}

\section{The New Codes}\label{sec3}

In this section we show how to guarantee the existence of cyclic
codes whose defining set contains only one $q$-coset containing at
least two consecutive integers. This fact produces conditions to
construct quantum codes with good parameters (in the sense of the
QSB). Theorem~\ref{teo1}, given in the following, is the main result
of this note.

\begin{theorem}\label{teo1}
Let $q\geq 3$ be a prime power and $n
> m$ be a positive integer such that $\gcd (q, n)=1$ and $\gcd (q^{a_i}-1, n)=1$
for every $i=1, 2, \ldots, r$, where $m={\operatorname{ord}}_n(q)
\geq r+2$ and $1 \leq r, a_1, a_2, \ldots , a_r < m$ are integers.
If $n|\gcd(t_2 , \ldots , t_r)$, where $t_j
=[(j-(j-1)q^{a_j}){(q^{a_j}-1)}^{-1}-{(q^{a_1}-1)}^{-1}]$ for every
$j=2,\ldots , r$ (the operations are performed modulo $n$), then
there exists an $[n, n- m^{*}, d\geq r+2]_{q}$ cyclic code, where
$m^{*}$ is the cardinality of the $q$-coset containing $r+1$
consecutive integers.
\end{theorem}
\begin{proof}
We want to investigate the following system of congruences
\begin{eqnarray*}
xq^{a_1}\equiv (x+1) \mod n\\
(x+1) q^{a_2} \equiv (x+2) \mod n\\
(x+2) q^{a_3} \equiv (x+3) \mod n\\
\vdots\\
(x+r-1) q^{a_r} \equiv (x+r) \mod n,\\
\end{eqnarray*}
where $1\leq r, a_1 , a_2, \ldots, a_r < m$. Since $\gcd (q^{a_i}-1,
n)=1$ for every $i=1, 2, \ldots, r$, it follows that the above
system is equivalent to
\begin{eqnarray*}
x \equiv {(q^{a_1}-1)}^{-1} \mod n \\
x \equiv (2- q^{a_2}){(q^{a_2}-1)}^{-1} \mod n\\
x \equiv (3- 2q^{a_3}){(q^{a_3}-1)}^{-1} \mod n\\
\vdots\\
x \equiv [r -(r-1)q^{a_r}]{(q^{a_r}-1)}^{-1} \mod n,\\
\end{eqnarray*}
where ${(q^{a_i}-1)}^{-1}$ denotes the multiplicative inverse of
$(q^{a_i}-1)$ modulo $n$.

The system has a solution if and only if
$$[j -(j-1)q^{a_j}]{(q^{a_j}-1)}^{-1} \equiv
[i -(i-1)q^{a_i}]{(q^{a_i}-1)}^{-1} (\mod n)$$ for all $i, j = 2,
\ldots , r$ and
$${(q^{a_1}-1)}^{-1} \equiv [i
-(i-1)q^{a_i}]{(q^{a_i}-1)}^{-1} (\mod n)$$ for all $i = 2, \ldots ,
r$. This means that
$$n|[(j-(j-1)q^{a_j}){(q^{a_j}-1)}^{-1}-{(q^{a_1}-1)}^{-1}]$$ for
every $j=2,\ldots , r$, i.e., $n|\gcd(t_2 , \ldots , t_r)$, where
$t_j =[(j-(j-1)q^{a_j}){(q^{a_j}-1)}^{-1}$ $-{(q^{a_1}-1)}^{-1}]$
for all $j=2,\ldots , r$.

Let $C$ be the cyclic code whose defining is the $q$-coset $C_{x}$.
From construction, the defining set of $C$, i.e., the coset $C_{x}$,
contains the sequence $ x, x+1, \ldots, x + r$ of $r+1$ consecutive
integers. From the BCH bound, the minimum distance $d$ of $C$
satisfies $d \geq r+2$. Since $|C_{x}|=m^{*}$, the dimension of $C$
equals $n - m^{*}$. Then, one can get an $[n, n - m^{*}, d\geq
r+2]_{q}$ code, as required.
\end{proof}

\begin{corollary}\label{cor1}
Let $q\geq 3$ be a prime power and $n
> m$ be a prime number such that $\gcd (q, n)=1$, where $m={\operatorname{ord}}_n(q)
\geq r+2$ and $1 \leq r, a_1, a_2, \ldots , a_r < m$ are integers.
If $n|\gcd(t_2 , \ldots , t_r)$, where $t_j
=[(j-(j-1)q^{a_j}){(q^{a_j}-1)}^{-1}-{(q^{a_1}-1)}^{-1}]$ for every
$j=2,\ldots , r$ and $a_1, a_2, \ldots , a_r$ are integers such that
$ 1\leq a_1 + a_2 + \ldots + a_r < m$ (the operations are performed
modulo $n$), then there exists an $[n, n- m^{*}, d\geq r+2]_{q}$
cyclic code.
\end{corollary}
\begin{proof}
Notice that since $n$ is prime, it follows that $\gcd (q^{a_i}-1,
n)=1$ for every $i=1, 2, \ldots, r$, because $a_1, a_2, \ldots , a_r
< m$. We next apply Theorem~\ref{teo1} and the result follows.
\end{proof}

Let $C_{x}$ be the $q$-coset of $x$. We denote by $C_{-x}$ the coset
of $-x$, where $-x$ is taken modulo $n$. With this notation we have:

\begin{theorem}\label{teo2}
Assume all the hypotheses of Theorem~\ref{teo1} hold. Let $C$ be the
cyclic code with defining set $C_{x}$, where $C_{x}$ is a coset
containing $r+1$ consecutive integers. If $C_{x}\neq C_{-x}$ then
there exists an $[[n, n-2m^{*}, d\geq r+2]]_{q}$ quantum code.
\end{theorem}
\begin{proof}
From~\cite[Lemma 1]{aly}, $C$ contains its (Euclidean) dual code
$C^{\perp}$. The dimension and the minimum distance of the
corresponding quantum code follow directly from Theorem~\ref{teo1}
and from Lemma~\ref{css}.
\end{proof}

\section{Examples and Code Comparison}\label{sec4}

\begin{example}\label{exa2}
Consider that $q=5$ and $n=11$; $m={\operatorname{ord}}_{11}(5)=5$.
The $5$-cosets are $C_0 = \{0\}$, $C_1 = \{ 1, 5, 3, 4, 9\}$ and
$C_2 = \{ 2, 10, 6, 8, 7\}$. If $C$ is the cyclic code with defining
set $C_1$, then it is a dual-containing code with parameters $[11,
6, d\geq 4]_{5}$. From Lemma~\ref{css}, one can get an ${[[11, 1,
d\geq 4]]}_{5}$ code. Similarly, take $q=17$ and $n=19$;
$m={\operatorname{ord}}_{19}(17)=9$. If $C$ is the code with
defining set $C_1 = \{ 1, 17, 4, 11, 16, 6, 7, 5, 9\}$ one has an
$[[19, 1, d\geq 5]]_{17}$ quantum code. We have an $[61, 56, d\geq
3]_{9}$ code with defining set $C_8 = \{8, 11, 38, 37, 28\}$. It is
a dual-containing code, so an $[[61, 51, d\geq 3]]_{9}$ quantum code
exists. There exists an $[67, 64, d\geq 3]_{29}$ dual-containing
code $C$ with defining set $C_{12} = \{12, 13, 42\}$. Hence, there
exists an $[[67, 61, d\geq 3]]_{29}$ quantum code. The existence of
an $[35, 31, d\geq 3]_{13}$ dual-containing code $C$ generates an
$[[35, 27, d\geq 3]]_{13}$ quantum code. An $[35, 31, d\geq 3]_{27}$
dual-containing code $C$ with defining set $C_3 = \{3, 11, 17, 4\}$
guarantees the existence of an $[[35, 27, d\geq 3]]_{27}$ quantum
code. An $[73, 70, d\geq 3]_{64}$ dual-containing code with defining
set $C_{21} = \{22, 21, 30\}$ exists, so there exists an $[[73, 67,
d\geq 3]]_{64}$ quantum code.
\end{example}

\begin{example}\label{exa3}
In this example, we construction cyclic codes whose defining set
consists of two $q$-cosets (the idea is the same as that presented
in Theorem~\ref{teo1}). An $[35, 27, d\geq 4]_{27}$ dual-containing
code $C$ with defining set consisting of $C_2$ and $C_3$ ensures the
existence of an $[[35, 19, d\geq 4]]_{27}$ quantum code. Taking the
cosets $C_{14} = \{14, 20, 30 \}$ and $C_{21}$  one has an $[[73,
61, d\geq 4]]_{64}$ code. Similarly, an $[[63, 51, d\geq 3]]_{11}$
code (coset $C_{43}$) and an $[[63, 39, d\geq 4]]_{11}$ code (cosets
$C_{43}$ and $C_{20}$) can be constructed. Analogously, an $[[63,
51, d\geq 3]]_{23}$ and an $[[63, 45, d\geq 4]]_{23}$ code (cosets
$C_4$ and $C_{27}$) can be constructed.
\end{example}

\begin{table}[!hpt]
\begin{center}
\caption{New quantum codes\label{table1}}
\begin{tabular}{|c |}
\hline Parameters of the new codes\\
\hline ${[[11, 1, d\geq 4]]}_{5}$\\
\hline $[[19, 1, d\geq 5]]_{17}$\\
\hline $[[35, 27, d\geq 3]]_{13}$\\
\hline $[[35, 27, d\geq 3]]_{27}$\\
\hline $[[35, 19, d\geq 4]]_{27}$\\
\hline $[[51, 35, d\geq 3]]_{32}$\\
\hline $[[61, 51, d\geq 3]]_{9}$\\
\hline $[[63, 51, d\geq 3]]_{11}$\\
\hline $[[63, 39, d\geq 4]]_{11}$\\
\hline $[[63, 51, d\geq 3]]_{23}$\\
\hline $[[63, 45, d\geq 4]]_{23}$\\
\hline $[[67, 61, d\geq 3]]_{29}$\\
\hline $[[73, 67, d\geq 3]]_{64}$\\
\hline $[[73, 61, d\geq 4]]_{64}$\\
\hline
\end{tabular}
\end{center}
\end{table}

It is usual the comparison of the new code parameters with the ones
presented in the literature. However, it seems that there is no
source available in literature for codes over large alphabets. More
precisely, the procedure adopted in~\cite{aly} does not generate
codes with relatively small length with respect to large alphabets.
In~\cite{laguardia:2014}, it is possible to derive good quantum
codes of minimum distance three only if the length is a prime
number, whereas here we can construct codes whose lengths are not
necessarily prime and with minimum distances greater than three.
Further, in~\cite{laguardia:2009}, only primitive codes were
constructed.

All quantum codes shown in Table~\ref{table1} seem to be new. Recall
that an $[[n, k, d]]_{q}$ quantum code satisfies $k + 2d \leq n + 2$
(QSB). Note that the new $[[67, 61, d\geq 3]]_{29}$ and $[[73, 67,
d\geq 3]]_{64}$ codes have parameters satisfying $n + 2 - k - 2d
\leq 2$; the parameters of the new ${[[11, 1, d\geq 4]]}_{5}$,
$[[35, 27, d\geq 3]]_{13}$ and $[[35, 27, d\geq 3]]_{27}$ codes
satisfy $n + 2 - k - 2d \leq 4$. The new $[[11, 1, d\geq 4]]_{5}$
code is comparable to the $[[17, 9, 4]]_{5}$ code shown in
\cite{edel}, and the new $[[61, 51, d\geq 3]]_{9}$ code is
comparable to the $[[65, 51, 4]]_{9}$ code shown in \cite{edel}.

\section{Final Remarks}\label{sec5}

We have constructed new quantum codes with good parameters by means
of the CSS construction. The existence of such codes are due to the
existence of suitable (classical) cyclic codes whose defining set
consists of only one $q$-coset which contains at least two
consecutive integers. This new method brings new ideas in order to
construct more new quantum (classical) cyclic codes.

\begin{center}
\textbf{Acknowledgements}
\end{center}
This research has been partially supported by the Brazilian Agencies
CAPES and CNPq.


\small

\begin{thebibliography}{99}
\bibitem{aly}
S.A. Aly and A. Klappenecker.
\newblock On quantum and classical BCH codes.
\newblock {\em IEEE Trans. Inform. Theory}, 53(3):1183--1188, 2007.

\bibitem{Calderbank:1998}
A.R. Calderbank, E.M. Rains, P.W. Shor, and N.J.A. Sloane.
\newblock Quantum error correction via codes over $GF(4)$.
\newblock {\em IEEE Trans. Inform. Theory}, 44(4):1369--1387, 1998.

\bibitem{edel}
Yves Edel. Table of quantum twisted codes. electronic address:
www.mathi.uni-heidelberg.de/~yves/Matritzen/QTBCH/QTBCHIndex.html


\bibitem{huffman03}
W.C. Huffman and V. Pless.
\newblock {\em Fundamentals of Error-Correcting Codes}.
\newblock Cambridge Univ. Press, New York, 2003.

\bibitem{Ketkar:2006}
A. Ketkar, A. Klappenecker, S. Kumar and P.K. Sarvepalli.
\newblock Nonbinary stabilizer codes over finite fields.
\newblock {\em IEEE Trans. Inform. Theory}, 52(11):4892--4914, 2006.

\bibitem{laguardia:2009}
G.G. La Guardia.
\newblock Constructions of new families of nonbinary quantum codes.
\newblock {\em Phys. Rev. A}, 80(4):042331(1--11), 2009.

\bibitem{laguardia:2011}
G.G. La Guardia.
\newblock New quantum MDS codes.
\newblock {\em IEEE Trans. Inform. Theory}, 57(8):5551--5554, 2011.

\bibitem{laguardia:2014}
G.G. La Guardia.
\newblock On the construction of nobinary quantum BCH codes.
\newblock {\em IEEE Trans. Inform. Theory}, 60(3):1528--1535, 2014.



\bibitem{Macwilliams:1977}
F.J. MacWilliams and N.J.A. Sloane.
\newblock {\em The Theory of Error-Correcting Codes}.
\newblock North-Holland, 1977.

\bibitem{Nielsen:2000}
M.A. Nielsen and I.L. Chuang.
\newblock {\em Quantum Computation and Quantum Information}.
\newblock Cambridge University Press, 2000.

\bibitem{Qian:2017}
J. Qian and L. Zhang.
\newblock Improved constructions for nonbinary quantum BCH codes.
\newblock {\em Int J Theor Phys}, DOI 10.1007/s10773-017-3277-y.
\end{thebibliography}
\end{document}